\documentclass[twocolumn,nofootinbib]{revtex4-2}
\usepackage{amsmath,amssymb,amsfonts}
\usepackage{algorithmic}
\usepackage{graphicx}
\usepackage{textcomp}
\usepackage{xcolor}
\usepackage{xspace}
\usepackage{physics2}
\usepackage{dsfont}
\usepackage{mathtools}
\usepackage{derivative}
\usepackage{amsthm}
\usepackage{comment}
\usepackage{xparse}
\usepackage{csquotes}
\usepackage{siunitx}
\usepackage{tikz}
\usepackage{qcircuit}
\usepackage[colorlinks=true,urlcolor=teal,
            linkcolor=teal,citecolor=teal]{hyperref}
\usepackage[capitalise]{cleveref}
\usepackage{todonotes}
\usepackage{booktabs}
\usepackage{orcidlink}
\usepackage{etoolbox}
\usepackage[caption=false]{subfig}
\usepackage{layouts}

\providetoggle{showtodos}
\settoggle{showtodos}{false}

\newcommand{\hide}[2]{\IfBooleanTF{#1}{}{#2}}

\usephysicsmodule{ab, ab.braket, ab.legacy}

\newcommand{\renyi}{Rényi\xspace}
\newcommand{\eg}{\emph{e.g.}\xspace}
\newcommand{\ie}{\emph{i.e.}\xspace}
\newcommand{\qft}{QFT\xspace}
\newtheorem{definition}{Definition}
\newtheorem{theorem}{Theorem}
\newtheorem{corollary}{Corollary}
\newtheorem{remark}{Remark}

\DeclareMathOperator{\ketzero}{\ket|0>}
\DeclareMathOperator{\ketone}{\ket|1>}
\DeclareMathOperator{\ketpsi}{\ket|\psi>}
\DeclareMathOperator{\ketpsiT}{\ket|\psi_T>}
\DeclareMathOperator{\kett}{\ket|t>}
\DeclareMathOperator{\ketb}{\ket|b>}
\DeclareMathOperator{\prbHam}{H_c}
\DeclareMathOperator{\geodeff}{\mu_\text{gd}}
\DeclareMathOperator{\Haar}{\text{Haar}}
\DeclareMathOperator{\probexp}{\mathds{E}}
\DeclareMathOperator{\bigO}{\mathcal{O}}
\DeclareMathOperator{\trace}{\text{tr}}
\DeclareMathOperator{\hilSpcn}{\mathcal{H}^{\otimes n}}
\DeclareMathOperator{\qntTrg}{\mathcal{T}}
\DeclareMathOperator{\qntBas}{\mathcal{B}}
\DeclareMathOperator{\bitFld}{\mathds{F}_2}
\DeclareMathOperator{\bitFldn}{\mathds{F}_2^n}
\DeclareMathOperator{\idm}{\mathds{1}}
\DeclareMathOperator{\sre}{\text{SRE}}
\DeclareMathOperator{\STAB}{\text{STAB}}
\DeclareMathOperator{\pliGrpn}{\mathcal{P}_n}
\DeclareMathOperator{\pliGrp}{\mathcal{P}}
\DeclareMathOperator{\pliFctGrpn}{\pliGrpn / \ab\langle \pm i \idm_n \rangle}
\DeclareMathOperator{\pliFctGrp}{\pliGrp / \ab\langle \pm i \idm \rangle}
\DeclareMathOperator{\clfGrp}{\mathcal{C}}
\DeclareMathOperator{\uniGrp}{\mathcal{U}}
\DeclareMathOperator{\px}{X}
\DeclareMathOperator{\py}{Y}
\DeclareMathOperator{\pz}{Z}
\DeclareMathOperator{\permOp}{\hat{\sigma}}
    
\DeclareMathOperator{\KLD}{D_\mathrm{KL}}    
\DeclareMathOperator{\expr}{\mathrm{Expr}}    
\DeclareMathOperator{\diag}{\mathrm{diag}}

\begin{document}

\title{Geometric and Resource-Theoretic Characterisation\\ of Non-Stabiliserness in Quantum Algorithms}

\author{Tom Krueger}\email{tom.krueger@othr.de}
\affiliation{Technical University of Applied Sciences Regensburg}%
\affiliation{FI CODE, Universität der Bundeswehr München}

\author{Wolfgang Mauerer}\email{wolfgang.mauerer@othr.de}
\affiliation{Technical University of Applied Sciences Regensburg}%
\affiliation{Siemens AG, Foundational Technologies}%

\date{\today}

\begin{abstract}
While there is strong evidence for advantages of quantum over classical computation, the repertoire of computational primitives  with 
proven or conjectured quantum advantage remains limited.  
A big challenge of quantum algorithmic design is a still incomplete understanding of the sources of quantum computational power.
Advancing towards systematic quantum advantage calls for a better understanding of the
efficient use of non-classical resources like non-stabiliser states. 

We present an approach to track non-classical contributions in the form of non-stabiliserness across various
algorithms by pairing resource theory of non-stabiliser entropies with the geometry
of quantum state evolution, and introduce permutation agnostic distance measures that
reveal and quantify non-stabiliser effects previously hidden by a subset of Clifford
operations. We find different efficiency in the use of non-stabiliserness for structured
and unstructured variational approaches, and show that greater freedom for classical
optimisation in quantum-classical methods increases unnecessary non-stabiliser consumption.
Our results open new means of analysing the efficient utilisation of quantum resources,
and contribute towards the targeted construction of algorithmic quantum advantage.
\end{abstract}

\maketitle

\section{Introduction}
Contrary to the general discussion of quantum versus classical computing that often
treats these as separate computational models, quantum computing (QC) can also be
seen as an extension to the classical computational model that adds new primitives
and resources. Quantum computations can (and for many suggested approaches also do,
particularly for any variational ansatz) contain classical parts~\cite{Bharti:2022,thelen:24:noisy-qaoa,franz:24:qce24},
which shifts the question of separating the two models to a more nuanced approach
of identifying inherently quantum parts in computations. While possible speed-ups
over purely classical approaches must obviously originate from quantum parts of a
computation, not every quantum sub-computation in proposed algorithms necessarily
needs to positively contribute to overall solution finding. Identifying reasons for
and structure of quantum speed-ups is a crucial question to improve the understanding
of chances and limitations of quantum computation. In this paper,  we address this
question from a novel point of view by using geometrical distance arguments within
a solution space. 

Several measures for quantumness have been established; the amount of entanglement
that manifests in computations is a prime candidate. It has been studied extensively
from the early times of quantum computing \cite{jozsa1997entanglement} into the present,
drawing substantial interest during the last few decades~\cite{preskill2012quantum,jozsa2003role,datta2007role,Mauerer2009}.
Entanglement is a distinct, perhaps the most non-classical aspect of quantum mechanics,
and is considered one of the fundamental resources of QC~\cite{khrennikov2021roots}.
However, its effect on computational power is not easy to characterise from a computer
science point of view. It is generally acknowledged and understood that entanglement
plays a fundamental role in many quantum algorithms and protocols, at least when using
the appropriate amount~\cite{Rohe:2025,Nakhl:2024,DiezValle:2021,Woitzik:2020,Gross:2009}.
Trying to pinpoint exactly where and how such non-classical advantage is exploited
necessitates more fine-grained insights. In particular, it is well known that
not all forms of entanglement are equal (or: equally useful)~\cite{odavic2024stabilizer}.
Even maximally entangled states like the seminal GHZ state can be prepared by Clifford
circuits; it is known that these can be efficiently simulated by a classical computer~\cite{Gottesman1998}.
States within the orbit of the Clifford group are called stabiliser states (STAB).
Conversely, states outside of STAB are referred to as non-stabiliser state (examples
of non-STAB entangled states include  W-states with three or more qubits) \cite{Garcia2014}.
Circuits required for their preparation are believed to be classically hard to simulate~\cite{Gottesman1998,Gottesman1997,Bravyi2005,Beverland:2020}.

Stabiliser-\renyi-Entropies (SRE) have been recently introduced to entropically measure
\emph{non-stabiliserness}, also referred as \emph{magic}, of quantum states \cite{Leone2022}.
We adopt SRE as measure of intermediate states to locate how and where non-classical
effects appear during the execution of contemporary quantum algorithms. 

The structure of the paper is as follows: In \cref{sec:relwork}, we review history
and significance of non-stabiliser resource theory and measures, particular stabiliser
\renyi entropies. Followed by \cref{sec:nonstab} where we provide an introduction
into the relevant definitions and characteristics of stabiliser \renyi entropies.
We then present a geometric perspective in \cref{sec:geopersp}, and show how to calculate
geodesic distances to target spaces by the means of taking the expectation value of
a special problem Hamiltonian. We then provide intuition for qubit permutation invariant
structures throughout the state evolution, from which we derive an extension of our
tools to work under invariance of the qubit
order, and show how this enables us to reveal previously overlooked non-stabiliser
effects. After that, we put our theoretic framework to use in \cref{sec:experiments}
to analyse the differences of intermediate non-stabiliser consumption in structured
and unstructured state evolutions. By combining the quantum resource theoretic $\sre$
measures with a geometric perspective, we are able to qualify the efficiency of non-stabiliser
consumption. We observed a significantly higher efficiency for the structured evolution
than for the unstructured case. We conclude and discuss the potential of combining
resource theoretic tools with geometric perspectives in \cref{sec:discussion,sec:perspecives}.

\begin{figure*}
    \centering
    \includegraphics[width=\textwidth]{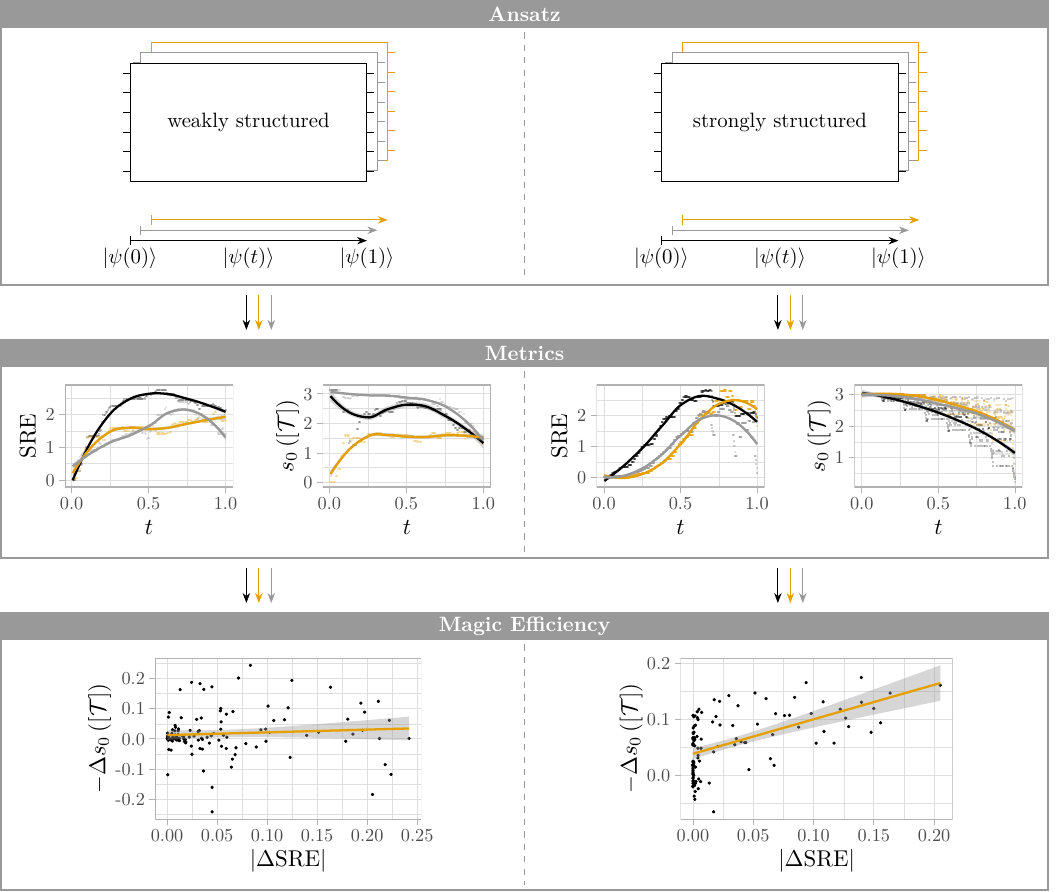}
    \caption{%
        Overview of our main methodology and contributions. In this paper we compare
        the state evolution $\ket|\psi(t)>$ (as $0 \overset{t}{\mapsto} 1$) throughout
        the circuit of a weakly and a strongly structured ansatz. Along this state
        evolution we compute two metrics. The first being the Stabiliser-\renyi-Entropy
        (SRE) and the second a geodesic distance from the target space invariant under
        qubit permutations. As can be seen in the middle pane, both metrics behave
        much more uniform in the strongly structured state evolution compared to the
        weakly structured one. In a final step (bottom pane) we compare the magic
        consumption $\abs{\Delta \sre}$ with the resulting geometric approach to the
        target $- \Delta s_0 \ab(\ab[\qntTrg])$. Doing so we can show that the strongly
        structured state evolution has a higher magic efficiency than the weakly structured
        one.
    }
    \label{fig:overview}
\end{figure*}

\section{Related Work}
\label{sec:relwork}
Given the significance of non-stabiliser effects in quantum 
computing, it is no surprise that understanding their 
properties and effects has been considered in numerous
contexts. In \cite{Gottesman1997}, Gottesman presented the 
\emph{stabiliser formalism} in his PhD thesis, laying the foundations of quantum error 
correction protocols. This formalism already covers many quantum specific phenomena like 
GHZ entanglement. The seminal Gottesman-Knill theorem~\cite{Gottesman1998}, presented 
shortly after, showed a quantum-classical separation by stating that every stabiliser 
circuit can be efficiently simulated by a classical computer. A stabiliser circuit is 
restricted to using gates from the Clifford group. An interesting conflict arises
as even stabiliser circuits are able to harness some quantum effects, yet they can
be efficiently simulated. This means the threshold to achieve quantum speedups must
be somewhere behind the obvious first line drawn between classical and non-classical
physics. For Clifford circuits, universality can
be recovered by aninjection process~\cite{Bravyi2005} where \emph{magic} states---non-stabiliser ancillas---serve
as consumable resources restoring universality by interfering with the Clifford part of the circuit.
Consequently, reaching a quantum advantage 
must be affected by magic state consumption. This motivated the the recent development of 
a resource theory for non-stabiliserness~\cite{Chitambar2019}, producing a diverse set of 
non-stabiliserness measures including stabiliser rank~\cite{Bravyi2016}, stabiliser 
fidelity~\cite{Bravyi2019}, or stabiliser nullity~\cite{beverland2020lower}. Following up 
on these results, more abstract characterisations of measures like non-stabiliser 
monotones have been defined~\cite{Haug2023}. The most relevant for our work are 
stabiliser \renyi entropies ($\sre$) as introduced in \cite{Leone2022}.
Stabiliser \renyi entropies are also known to be monotones for non-stabiliserness
resource theory \cite{Leone2024}. While most measures of non-stabiliserness are  hard
to compute, $\sre$ can be efficiently determined for low entanglement systes can be
represented as matrix product states~\cite{Oliviero2022matrixprod}. Further, $\sre$s can also be 
determined empirically through measurements~\cite{Oliviero2022}.

\section{Non-Stabiliserness}
\label{sec:nonstab}
Before defining a measure for non-stabiliserness, it seems pertinent to take a brief moment for discussing the term \emph{non-stabiliserness}. We already mentioned that STAB is given by the orbit of the Clifford group (recall that the orbit of an element \(x\) in a group \(G\) is given by \(G(x) \coloneq \{gx \in G: g \in G\}\)). The Clifford group $\clfGrp_n = \ab\{V \in \uniGrp_{2^n} : V \pliGrpn V^\dag = \pliGrpn\}$ is the normaliser of the Pauli group $\pliGrpn = \ab\langle X, Y, Z\rangle^n$, where $X, Y, Z$ are the Pauli operators and $n$ denotes the number of qubits. In the following, we drop suffix $n$ if the number of qubits can be deduced from the context, or to describe systems of arbitrary (but finite) size.

\begin{definition}[see Ref.~\cite{Leone2022}]
    The $\alpha$-Stabilizer-\renyi-Entropy is defined as:
\label{def:SRE}
\begin{align}
\sre_\alpha(\ket|\psi>) &= \frac{1}{1 - \alpha} \log \sum_{P \in \pliFctGrpn}  \Xi_P^\alpha\ab(\ket|\psi>)  - \log 2^n\label{eq:SRE}\\
    \Xi_P \ab(\ket|\psi>) &= \frac{1}{2^n} \braket<\psi|P|\psi>^2\label{eq:XiP}
\end{align}
\end{definition}

\Cref{def:SRE} may require some explanation to establish an intuitive understanding. Let us start from \cref{eq:XiP}. Note that $\Xi_P \ab(\ket|\psi>) \leq 1$ and $\sum_{P \in \pliFctGrpn} \Xi_P \ab(\ket|\psi>) = 1$. Thus, $\ab\{\Xi_P\}_{P \in \pliFctGrpn}$ induces a probability distribution on a state $\ket|\psi>$. Here we also immediately see why we used the factor group of $\pliGrpn$, ignoring the scalar unit factors $\pm 1$ and $\pm i$. Due to the applied square in \cref{eq:XiP}, they would only double up in the sum, not contributing any valuable information to the distribution. This also explains the normalisation factor of $2^{-n}$ corresponding to $\ab|\pliFctGrpn| = 2^n$ in contrast to the perhaps expected $4^n = \ab|\pliGrpn|$. We also see, with $\ab\{\Xi_P\}$ being a probability distribution, \cref{eq:SRE} simply defines a family of \renyi entropies offset by $\log 2^n$.

To demonstrate and provide intuition about how Stabiliser-\renyi-Entropies work, it
is worth looking into their main characteristics, to understand how $\sre$s characterise
stabiliser states. For this, we will revisit the property which is most important
for our work, namely that non-STAB states are characterised by a non-zero $\sre$.

Let $\ket|\psi> \in \STAB$ be some not further specified stabiliser state. Then $\ket|\psi>$
is in the Clifford orbit of $\ket|0>$, meaning that there exist a $U \in \clfGrp$
such that $U \ketzero = \ketpsi$. Due to the Clifford group stabilising the Pauli
group, we have $U^\dag P_j U = P_j$ for $P_i, P_j \in \pliFctGrpn$. In fact, $\clfGrp$
is isomorphic to the group of permutations in a sense that $U^\dag P_j U : P_j \mapsto
P_{\pi(j)}$. Therefore, $\ab\{\Xi_{P_j}\ab(\ketpsi)\} = \ab\{\Xi_{P_{\pi(j)}}\ab(\ketzero)\}$
and consequently $\sum_{P \in \pliFctGrp} \Xi_P^\alpha \ab(\ketpsi) = \sum_{P \in
\pliFctGrp} \Xi_P^\alpha \ab(\ketzero)$. Note that, $\eval{\braket<a|\pz|a>}_{a =
0,1} = \pm 1$, $\eval{\braket<a|\idm|a>}_{a = 0,1} = 1$ and $\eval{\braket<a|\px,
\py|a>}_{a = 0,1} = 0$, which leads to the conclusion that
\begin{equation}
    \braket<0|P|0> = \begin{cases}
        0 & \text{if} \quad \exists i : \sigma_i \in \{\px, \py\} \\
        1 & \text{otherwise}
    \end{cases}
\end{equation}
for all $P = \sigma_i \otimes \cdots \otimes \sigma_n \in \pliFctGrpn$. There are
$2^n$ many $P \in \pliFctGrpn$ such that $\braket<0|P|0> = 1$. As a result $\sum_{P
\in \pliFctGrpn} \Xi_P^\alpha \ab(\ketpsi) = 2^n 2^{-n\alpha} = 2^{n(1-\alpha)}$.
Here is where the offset of $\log(2^n)$ in \cref{eq:SRE} comes into play, as $\sre_\alpha
\ab(\ketpsi) = (1 - \alpha)^{-1} \log 2^{n(1-\alpha)} - \log 2^n = \log 2^n - \log
2^n$. We conclude that $\sre_\alpha \ab(\ketpsi) = 0$ for all $\ketpsi \in \STAB$.

Now we will investigate the other direction, for which we will use an alternative
characterisation of stabiliser states, which is that $\ketpsi$ is in $\STAB$ if and
only if there exists a subset $S \subset \pliGrpn$ such that $\ab|S| = 2^n$ and $A
\ketpsi = \ketpsi$ for all $A \in S$. Now let's assume $\sre_\alpha \ab(\ket|\psi>)
= 0$ for some arbitrary state $\ket|\psi>$. Written out, that gives us $(1-\alpha)^{-1}
\log \sum_{P \in \pliFctGrp} \Xi_P^\alpha \ab(\ketpsi) - \log 2^n = 0$ or rewritten
$\log \sum_{P \in \pliFctGrp} \Xi_P^\alpha \ab(\ketpsi) = \log 2^{n(1 - \alpha)}$.
Thus, $\log \sum_{P \in \pliFctGrp} 2^{-n\alpha} \braket<\psi|P|\psi>^{2\alpha} =
2^{n - n\alpha}$. From this we can derive a condition on $\ketpsi$ for $\sre_\alpha
\ab(\ketpsi)$ to equal to 0:
\begin{equation}
    f\ab(\alpha) = 2^n
\end{equation}
where $f\ab(\alpha) = a_1\ab(\alpha) + \cdots + a_{2^n}\ab(\alpha)$ with $a_i = \braket<\psi|P_i|\psi>^{2\alpha}$
and $P_i \in \pliFctGrpn$. Due to $f$ being a constant function, we have $\odv{}{\alpha}
f = 0$. Additionally, we know that all $a_i \geq 0$ and therefore $\odv{}{\alpha}\braket<\psi|P|\psi>^{2
\alpha} = 0$ for all $P \in \pliFctGrpn$. From, $\odv{}{\alpha} \braket<\psi|P|\psi>^{2\alpha}
= 2\braket<\psi|P|\psi>^{2\alpha} \log \braket<\psi|P|\psi> = 0$ we conclude, that
$\braket<\psi|P|\psi> \in \{0, 1\}$ for all $P \in \pliFctGrpn$. Note that $\braket<\psi|P|\psi>
= 1$ only if $P \ketpsi = \ketpsi$ and $f(\alpha) = 2^n$. Thus, there exists a subset
$S \subset \pliFctGrpn$ such that $\ab|S| = 2^n$ and $A \ketpsi = \ketpsi$ for all
$A \in S$; showing that $\ketpsi \in \STAB$. We arrive at the following key characteristic
of the $\sre$ measure proven in its introductory paper \cite{Leone2022}.

\begin{theorem}[originally proven in \cite{Leone2022}]
    A state $\ket|\psi>$ is in STAB if and only if $\sre_\alpha \ab(\ket|\psi>) = 0$.
\end{theorem}

We refer readers to the original paper for the original proof and a more in-depth
discussion~\cite{Leone2022}. Also note following observation, which makes $\sre$ such
an interesting measure for identifying non-classical computation steps in quantum
circuits.

\begin{corollary}
    Stabiliser-\renyi-Entropies are invariant under Clifford operations.
\end{corollary}

\begin{proof}
    This follows directly from the isomorphism between the Clifford group and permutations. Let $\ketpsi$ be an arbitrary state and $U \in \clfGrp$, then we get that $\ab\{\Xi_{P_j} (U\ketpsi)\} = \ab\{\Xi_{P_{\pi(j)}} (\ketpsi)\}$ and thus $\sre_\alpha \ab(U\ketpsi) = \sre_\alpha \ab(\ketpsi)$.
\end{proof}

\section{Geometric Perspective}
\label{sec:geopersp}

If we take a random Haar sampled state $\ketpsi \sim \Haar$. Then the expected $\sre$
is $\probexp_{\ketpsi \sim \Haar}\ab(\sre_\alpha \ab(\ketpsi)) \in \bigO (n)$ for
all $\alpha \geq 2$, with overwhelming probability~\cite{Gu2024}. Additionally, $\sre$s
are linear upper bounded by $\sre_\alpha \ab(\ketpsi) \leq \log(2^n) \in \bigO(n)$~\cite{Leone2022}
(technically more precise, stabiliser entropy scales linearly with concentration bounds).
This means, although intermediate states with $\sre \ge 0$ are linked with and even
necessary for quantum advantage, their occurrence is nothing special and has to be
expected. Consequently, this raises the question if observed non-stabiliserness contributes
to the computation, or if it is merely a by-product of a suboptimal choice of unitary
propagators. To make this notion more precise, we need to clarify the meaning of \emph{contributing
to the computation} mean: Every computational process can be interpreted as a state
evolution that drives a specific initial state $\ket|\psi_0>$ to a target state $\ket|\psi_T>$
that encodes a problem solution (or a superposition thereof). Geometrically speaking,
such a state evolution resembles a rotation of the state vector. The whole circuit
represents one singular unitary, which in turn corresponds to a direct rotation from
the initial to the final state around the rotational axis defined by said unitary.
This only applies from the most top-level view and discard the actual realisation
of the circuit's unitary given by a concrete partitioning into quantum gates and their
correct sequencing. The circuits gate level realisation induces a path of the resulting
state evolution, which is most likely diverging from the shortest path at some point.
In geometric terms, the shortest, most direct path of this state evolution would be
characterised by the geodesic from the initial state to the target. In \cite{Anandan1990}, Anandan and Aharonov presented exactly this geometric perspective in conjunction
with the concept of \emph{geodesic efficiency} $\geodeff = s_0 / s$ of a state evolution
where $s_0$ is the geodesic distance and $s$ the actual distance travelled. If we
want to specify the distance to a specific state $\ket|\phi>$, we write $s_0 \ab(\ket|\phi>)$
and $s \ab(\ket|\phi>)$ and if the initial state is not clear from the context we
write $s \ab(\ketpsi, \ket|\phi>)$ and $s_0 \ab(\ketpsi, \ket|\phi>)$

\subsection{Problem Hamiltonian}
Usually there is more than one unique solution to a computational problem, for instance
all binary variable assignments satisfying a propositional satisfiability problem.
This adds variability to the geometric perspective discussed above. Instead of rotating
our initial state to a specific target state $\ketpsiT$, a quantum algorithm has the
freedom to reach any state within the \emph{target space}, which is the subspace of
$\hilSpcn$ that contains all superpositions of quantum states encoding problem solutions.
Thus, the geodesic distance $s_0$ from above needs to be reinterpreted to be the shortest
geodesic distance to one of the states in the target space. In the following, we will
address this by first defining the target space based on an indicator function of
problem solutions and a two-level problem Hamiltonian projecting on said target space.
In contrast to Hamiltonians usually found in combinatorial optimisation such two-level
Hamiltonians encode solution spaces of decision problems, where the solutions are
not weighted and the task is to find any valid solution. We then show, that the expected
value of this problem Hamiltonian corresponds to the scaled inverse geodesic distance
from the initial state to the target space.

\begin{definition}
\label{def:qntTrg}
    Let be $c : \bitFldn \rightarrow \bitFld$ the solution verifier of a problem with
    a finite set of classical solutions $T = \ab\{t \in \bitFldn : c\ab(t) = 1\}$.
    We then extend $c(\cdot)$ to quantum states $\ketpsi = \sum_{b \in \bitFldn}
    \alpha_b \ketb$ as a linear function on the basis states $\{\ket|b>\}_{b \in \bitFldn}$:
    \begin{equation}
        c\ab(\ketpsi) = \sum_{b \in \bitFldn} \ab|\alpha_b|^2 c\ab(b)
    \end{equation}
    Further, we define a quantum target space
    \begin{equation}
        \label{eq:qntTrg}
        \qntTrg = \ab\{\kett : c\ab(\kett) = 1\} \subset \hilSpcn
    \end{equation}
\end{definition}

\begin{remark}
    Note that $\qntTrg$ is indeed a complete subspace of $\hilSpcn$, spanned by $\ab\{\ketb : b \in \bitFldn , c\ab(b) = 1\}$. Thus, $\qntTrg$ has a dimension of $\ab|T|$.
\end{remark}

\begin{definition}
    \label{def:prbHam}
    Based on $c(\ketpsi)$, we define a 2-level problem Hamiltonian $\prbHam$ by the condition that
    \begin{equation}
        \label{eq:prbHam}
        \braket<\prbHam> = c\ab(\ketpsi)
    \end{equation}
\end{definition}
$H_c$ is a projector onto $\qntTrg$ that can explicitly defined by $\prbHam = \sum_{t \in T} \ketbra|t><t|$.

\begin{theorem}
    \label{the:expPrbHamMinS0}
    Given a target space $\qntTrg$ and the corresponding problem Hamiltonian $\prbHam$ according to \cref{def:qntTrg,def:prbHam}, we have
    \begin{equation}
        \label{eq:expPrbHamMinS0}
        s_0\ab(\qntTrg) := \min_{\kett \in \qntTrg} s_0\ab(\kett) = 2 \arccos \braket<\prbHam>
    \end{equation}
\end{theorem}

\begin{proof}
    Let $\qntBas_{\qntTrg} = \ab\{\kett : t \in T\}$ the basis of $\qntTrg$. We then expand $\qntBas_{\qntTrg}$ with $\qntBas_{\overline{\qntTrg}} = \ab\{\ketb : b \in \bitFldn \setminus T\}$ such that $\qntBas_{\qntTrg} \cup \qntBas_{\overline{\qntTrg}}$ forms a basis of $\hilSpcn$.  Given that expanded basis, we can write every state $\ketpsi \in \hilSpcn$ as $\ketpsi = \sum_{i = 1}^{\ab|T|} \tau_i \ket|t_i> + \sum_{i = 1}^{n - \ab|T|} \beta_i \ket|b_i>$, with $\sum_{i = 1}^{\ab|T|} \ab|\tau_i|^2 + \sum_{i = 1}^{n - \ab|T|} \ab|\beta_i|^2 = 1$ and for all $\kett \in \qntTrg$ we have $\sum_{i = 1}^{\ab|\qntTrg|} \ab|\tau_i|^2 = 1$. Now, let $\prbHam$ be a problem Hamiltonian as defined in \cref{def:prbHam}, then $\prbHam = \sum_{\kett \in \qntBas_{\qntTrg}} \ketbra|t><t|$ and
    $$
        0 \leq \braket<H_c> = \sum_{i = 1}^{\ab|T|} \ab|\tau_i|^2 \leq 1
    $$
    Now, let's take an arbitrary state $\ketpsi$, then $\braket<\prbHam> = \braket<\psi|\prbHam|\psi>$ is exactly the overlap between $\ketpsi$ and its projection onto the target space $\prbHam \ketpsi$, which satisfies
    $$
        \braket<\prbHam> = \max_{\kett \in \qntTrg} \abs{\braket<\psi|t>}
    $$
    Now we use that $s_0\ab(\kett) = 2 \arccos \abs{\braket<\psi|t>}$ \cite{cafaro2025}. Now due to the monotonicity of $\arccos$ in $[0,1]$ we can pull out the $\max$ from $2 \arccos \braket<\prbHam>$ to end up at \cref{eq:expPrbHamMinS0}.
\end{proof}

\subsection{Permutations}
\begin{figure}[htbp]
    \subfloat[QAOA (7 qubits, 7 layers)]{\includegraphics{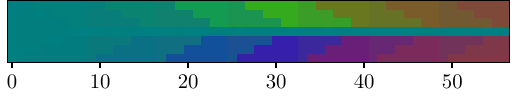}}

    \subfloat[QFT (10 qubits)]{\includegraphics{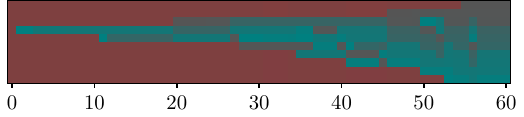}}
    
    \caption{%
        Evolution of the colour representation of the state in quantum circuits. Every
        vertical slice at $x=i$ represents the colour spectrum of the state after
        the $i$th gate. The reduced one qubit density matrices are mapped to a hue-saturation-value
        colour with $\text{hsv}\ab(\braket<P_0>, \braket<P_{+}>, \braket<P_{+i}>)$.
        Within a vertical slice, they are sorted according to the $\text{hsv}$ tuple.
        This reveals stark differences in the evolution of qubit permutation invariant
        structures throughout the state evolutions driven by QAOA and QFT circuits.
    }
    \label{fig:colorrep}
\end{figure}

In a typical quantum circuit, qubits are sequentially numbered. This numbering implies
an unsubstantiated sense of order of qubits in the quantum circuit and state vector
picture. Indeed, it is actually completely
arbitrary and nothing more of a naming convention. Qubit $q_i$ and $q_j$ could also
be remapped $q_{\sigma(i)}$ and $q_{\sigma(j)}$ for some permutation $\sigma \in S_n$.
Yet, geometrically the state vectors of $\ket|0010>$ and $\ket|0001>$ are orthogonal.
In this section we will extend the geometric distance measure to address this dilemma.
From a computational standpoint, looking at a concrete problem instance there is no
order of variables by naming, 
there are relations between variables\footnote{Which then could enforce an order not
necessarily present in the naming of variables}. Variables are then being mapped to
qubits and inter variable connections usually result in qubit interactions \ie multi
qubit gates. Those connections
could better be represented as a graph, which in turn is isomorphic under vertex permutations.
On this graph, we assign to each vertex a colour based on properties of its linked
qubit. By keeping track of the necessary permutations, we can, at each intermediate
time, determine the permutation order of qubits based on their assigned colour.

One could thus visualise a quantum state evolution as the change of a colour spectrum
through time. Solving a classical problem, we are basically interested in the measurement
probabilities of all qubits and the resulting bit-string, hopefully encoding a possible
solution to the problem. We therefore exemplary map $\braket<P_0>$, $\braket<P_+>$
and $\braket<P_{+i}>$ to the hue, saturation and value component of an HSV colour.
Here $\braket<P_0>$, $\braket<P_{+}>$ and $\braket<P_{+i}>$ are the probabilities
of the reduced mixed system of said qubit being in the state $\ket|0>$, $\ket|+>$
and $\ket|+i>$. Now the, qubits can be ordered according to their hue. \Cref{fig:colorrep}
demonstrates how this representation, which is qubit permutation invariant, still
reveals highly specific structures of quantum state evolutions. As we will show below,
the ordering does not alter non-stabiliserness, as it can be performed by an efficient
Clifford circuit. Therefore, it can be ignored regarding our analysis of non-stabiliserness
resource consumption. Additionally, introducing, at the worst case, one ordering and
reordering before and after each computational step does also not change the complexity
theoretic characterisation of the circuit, as it, given the presumption of a polynomial
sized initial circuit, only adds a polynomial amount of permutation circuits which
themselves also only have a polynomial complexity. Therefore, questions regarding
the link between non-stabiliser consumption and quantum advantages can be investigated
with frameworks factoring out permutational degrees of freedom.

Our colour spectrum representation of states served as a visual intuition for qubit
permutation invariant similarities between states. As an intermediate step, we will
now reformulate this idea in a more mathematical way before defining a rigorous permutation
invariant distance measure. Consider $\text{hsv}\ab(\braket<P_0>, \braket<P_+>, \braket<P_{+i}>)$
as a point in a colour space, then the color spectrum state representation is a set
of points in this space, where each point corresponds to the vector of Pauli expectations
$\braket<P_0>, \braket<P_+>, \braket<P_{+i}>$ for a specific qubit of the state. Given
two states, we get two clouds of $n$ points in the colour space. The similarity of
both clouds can be computed by matching points of both clouds in pairs. The pairwise
matching cost will be the euclidean distance between Pauli expectation vectors of
both paired states. Finally, we get the similarity measure by taking the mean pairtwise
matching cost of the minimal matching. Informally speaking, we calculated how far
the points of both clouds drifted apart on average. Mathematically, we define a function
$f : \mathds{N} \times \hilSpcn \to \mathds{R}^3$ such that 
\begin{equation}
    f\ab(j, \ketpsi) = \ab(\trace(\rho_j P_0), \trace(\rho_j P_+), \trace(\rho_j P_{+i})),
\end{equation}
with $\rho_j$ being the reduced density matrix of the $j$-th qubit. Using $f$ we now
define a state map $C: \hilSpcn \to \mathds{R}^{n \times 3}$:
\begin{equation}
    C(\ketpsi) = \ab[f\ab(i, \ketpsi)]_{j = 1}^n
\end{equation}
On this state representation we calculate a distance $d_{\text{MMC}} (\ket|\psi_1>,
\ket|\psi_2>)$ which is the mean of the minimal bipartite row matching between $C(\ket|\psi_1>)$
and $C(\ket|\psi_2>)$ with matching costs $\alpha_{j,l} = \norm{f(l, \ket|\psi_2>) - f(j, \ket|\psi_1>)}$. 

\Cref{fig:qftGeoDist} shows how this neatly bridges between the colour spectrum intuition
and the rigorous permutation invariant distance we will define now. As a first step,
we define a permutation operator capturing the notions discussed above.

\begin{definition}
    \label{def:permOp}
    Given a permutation $\sigma \in S_n$ and $\ketb \in \qntBas^n$ where $\ketb =
    \ket|b_1> \otimes \cdots \otimes \ket|b_n>$ with $\ket|b_i> \in \ab\{\ketzero,
    \ketone\}$, then we define $\hat{\sigma} \in U(\hilSpcn)$ as
    \begin{equation}
        \hat{\sigma} \ketb = \ket|b_{\sigma(1)}> \otimes \cdots \otimes \ket|b_{\sigma(n)}>
    \end{equation} 
    and 
    \begin{equation}
        \hat{\sigma} \sum_{b \in \qntBas^n} \alpha_b \ketb = \sum_{b \in \qntBas^n} \alpha_b \hat{\sigma} \ketb.
    \end{equation}
\end{definition}

Note that in \cref{def:permOp} the inverse operator $\hat{\sigma}^\dag \in U(\hilSpcn)$
corresponds to the inverse permutation $\sigma^{-1} \in S_n$. Next, we have to show
that the $\sre$ measure is invariant under such permutation operators.

\begin{theorem}
\label{the:sreInvUnderPerm}
Let $\hat{\sigma}$ be a permutation operator as defined in \cref{def:permOp}, then
$$
\sre_\alpha \ab(\hat{\sigma} \ketpsi) = \sre_\alpha \ab(\ketpsi)
$$    
\end{theorem}

\begin{proof}
    Every permutation $\sigma \in S_n$ can be decomposed into a sequence of 2-cycles, which can be realised by a single swap gate. Thus, $\hat{\sigma} \in \hilSpcn$ can be realised by a sequence of swap gates, which are Clifford operations. Since, $\sre_\alpha$ is invariant under Clifford operations, it also is for all permutation operators constructed as defined in \cref{def:permOp}.
\end{proof}

Now, after we have formalised the idea of invariance under permutation on the operational side, we will do the same for the objects of interest. We do this by subsuming all states equal under permutation into equivalence classes and then extend this to the target space itself.

\begin{definition}
    \label{def:equivRel}
    We define an equivalence relation $\ket|\psi_l> \sim \ket|\psi_r>$ which is satisfied if and only if there exist a permutation operator $\hat{\sigma}$ as defined in \cref{def:permOp}, such that $\ket|\psi_r> = \hat{\sigma} \ket|\psi_l>$.
    Then 
    \begin{equation}
        \ab[\ketpsi] = \ab\{\ket \hat{\sigma} \ket|\psi> : \forall \hat{\sigma}\}
    \end{equation}
    is the corresponding equivalence class of $\ketpsi$ under $\sim$ and further $\sre\ab(\ab[\ketpsi]) = \sre\ab(\ketpsi)$.
    Let $\qntTrg$ be a subspace of $\hilSpcn$, we then extend this notion by defining
    \begin{equation}
        \ab[\qntTrg] = \bigcup_{\kett \in \qntTrg} \ab[\kett]
    \end{equation}
\end{definition}

From \cref{the:sreInvUnderPerm} it also immediately follows that $\sre_\alpha\ab(\ketpsi) = \sre_\alpha\ab(\ab[\ketpsi])$. For the geodesic distance, we need to extend the definition to equivalence classes.

\begin{definition}
\label{def:disEqivCls}
Let $\ab[\ket|\phi>]$ be a equivalence class of states then
    \begin{equation}
        s_0\ab([\ket|\phi>]) = \min_{\ket|\phi'> \in \ab[\ket|\phi>]} s_0 \ab(\ket|\phi'>)
    \end{equation}
For $\ab[\qntTrg]$ we extend $s_0$ in a similar fashion to
    \begin{equation}
        s_0\ab([\qntTrg]) = \min_{\kett \in \qntTrg} s_0 \ab(\ab[\kett])
    \end{equation}
\end{definition}

Determining the distance $\ab[\qntTrg]$ requires tracing all permutations of all possible solution states, which can be a bit tricky. Lucky, we can show that the distance from $\ketpsi$ to $\qntTrg$ is equal to the distance from $\ab[\ketpsi]$ to $\qntTrg$.

\begin{theorem}
    Given a target space permutation equivalence class $\ab[\qntTrg]$ as defined in \cref{def:equivRel}, it holds that
    \begin{equation}
        s_0\ab([\qntTrg]) = s_0\ab(\ketpsi, [\qntTrg]) = s_0 \ab(\ab[\ketpsi], \qntTrg)
    \end{equation}
\end{theorem}

\begin{proof}
By definition, we have that $s_0 \ab(\ab[\ketpsi], \qntTrg) = \min_{\ket|\psi'> \in [\ketpsi]} s_0\ab(\ket|\psi'>, \qntTrg)$ which equals  $\min_{\permOp} s_0\ab(\permOp \ketpsi, \qntTrg) = \min_{\permOp}2 \arccos \braket<\psi|\permOp^\dag \prbHam \permOp| \psi>$. As we are minimising over the whole group of all permutation operators we can also minimise over all complex conjugate operators instead $\min_{\permOp^\dag}2 \arccos \braket<\psi|\permOp \prbHam \permOp^\dag |\psi>$. By the canonical definition of $\prbHam$ we have $\permOp \prbHam \permOp^\dag = \sum_{t \in T} \permOp \ketbra|t><t| \permOp^\dag$. Recall that $\ab\{\kett : t \in T\}$ is the basis of the corresponding quantum target space $\qntTrg$. This means, by applying $\permOp \prbHam \permOp^\dag$ we are performing a basis transformation on the target space, measuring the expected probability of $\ketpsi$ being in he permuted target space. By minimising over all permutations we get $s_0 \ab(\ketpsi, \ab[\qntTrg])$, thus in conclusion $s_0\ab(\ketpsi, \ab[\qntTrg]) = s_0 \ab(\ab[\ketpsi], \qntTrg)$.
\end{proof}

\begin{figure}[htbp]
    \centering
    \includegraphics{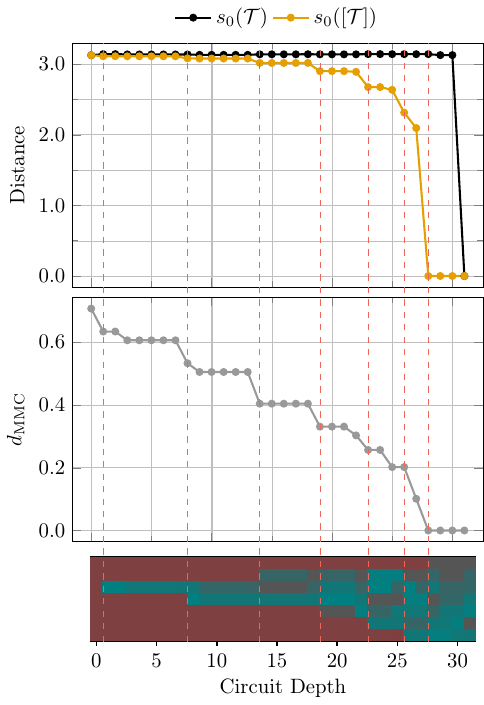}
    \caption{Minimal geodesic distance for increasing circuit depths. Non-\hspace*{0mm}Clifford
    computational progress can be observed prior to the final qubit order reversal
    when all target space permutations $s_0 \ab(\ab[\qntTrg])$ (ochre) are considered.
    As such effects are not visible in the direct distance to the target space $s_0
    \ab(\qntTrg)$ (black) that  neglects permutations, this demonstrates how potential
    non-stabiliser effects can be masked by non-\hspace*{0mm}Clifford-\hspace*{0mm}agnostic
    measures (lines are used to guide the eye and have no significance).%
    \newline
    The permutation invariant distance $s_0 \ab(\ab[\qntTrg])$ is also much more inline
    with the structural chances observed in the color spectrum representation of the
    state evolution under the QFT circuit. See \cref{fig:colorrep} for a detailed
    description of the color spectrum representation.
    \newline
    Also, the mean minimal matching cost distance ($d_{\text{MMC}}$) to the final
    circuit state aligns with both $s_0 \ab(\ab[\qntTrg])$ and the structural changes
    of the color spectrum representation.}
    \label{fig:qftGeoDist}
\end{figure}

\begin{figure}[htbp]
    \includegraphics[width=\linewidth]{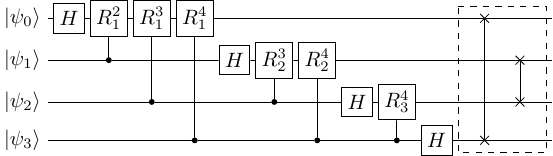}
    \caption{\qft circuit with four qubits. The dashed box marks the qubit order inversion block of swap gates. Non-stabiliser computations take place before
    this block, but their computational influence on the geodesic distance is masked by the final qubit reordering.}
    \label{fig:qftCirc}
\end{figure}

The quantum Fourier transform (\qft) is a good example to demonstrate this effect,
and additionally shows how to  calculate the distance of the closest target space
permutation. As of today, the \qft is regarded as \emph{the} seminal primitive contributing
to quantum advantage, finding application in a wide range of quantum algorithms reaching
proven quantum speed-ups \cite{Kitaev1995,Shor1997}, thus making it an interesting
quantum primitive to study \cite{Linden2022}. If one intends to use a distance measure
to quantify computational progress it must be applicable to analyse the QFT as well.
The interesting parts of the \qft circuit take part before the
qubits are reordered in a final step (marked section in \cref{fig:qftCirc}). This
is problematic when looking at distance measures based on state to space overlaps
like the geodesic distance. The block of swap gates implementing the reordering is
entirely Clifford, yet looking at the geodesic instance $s_0 \ab(\qntTrg)$ one could
be under the impression that all the computational progress takes place in this section
of the circuit. This cannot be the case, a the \qft algorithm enables exponential
speed-ups. Therefore valuable computational progress has to be made before, taking
possible target space permutations into account to reveal such effects (see \cref{fig:qftGeoDist}).

\section{Numerical Simulations}\label{sec:experiments}
\begin{figure*}[htbp]
\subfloat[Properties of intermediate states (\(x\)-axis) evolving under an \textbf{weakly structured}
    ansatz for different instances (\(y\)-axis).  \emph{Top:} The state approaches
    the target space with notable erraticity, as evidenced by irregular fluctuations
    in the geodesic distance from {\([\qntTrg]\)}. While non-stabiliserness (\emph{middle})
    appears to evolve comparatively smooth, the variation in resource consumption,
    as quantified by  \(|\Delta\text{SRE}|\) (\emph{bottom}), continues to display
    a markedly irregular behaviour. \label{fig:vqaheatmaps}%
]{\includegraphics{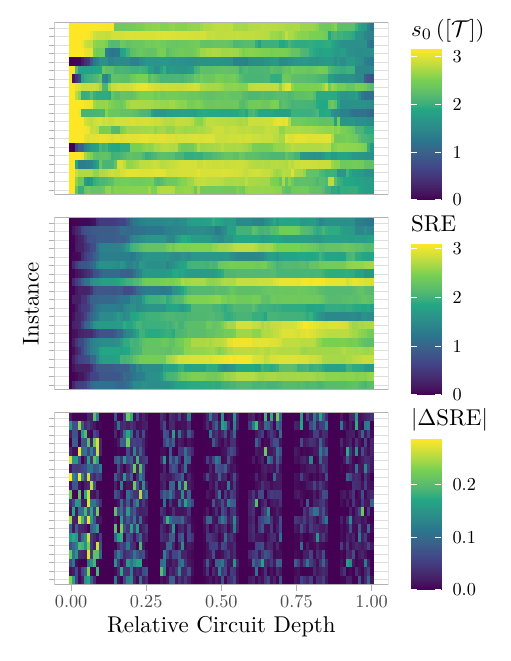}}\hfill
\subfloat[
    Properties of intermediate states (\(x\)-axis) evolving under a \textbf{strongly structured}
    ansatz for different instances (\(y\)-axis). \emph{Top:} The state approaches
    the target space {\([\qntTrg]\)} smoothly. The accumulation of non-stabiliserness
    (\emph{middle}) follows a structured trajectory and reaches its apex at about
    75\% relative circuit depth, after which it gradually diminishes. This behaviour
    is mirrored in the patterns of non-stabiliser resource consumption (\emph{bottom}).\label{fig:qaoaheatmaps}%
]{\includegraphics{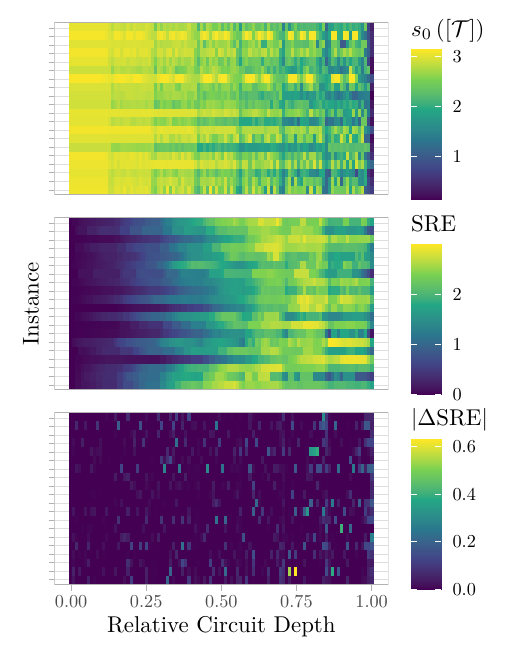}}
\caption{%
    Comparison of intermediate geodesic distances $s_0 \ab(\ab[\qntTrg])$, non-stabiliserness
    $\sre$ and non-stabiliser consumption $\abs{\Delta \sre}$ between unstructured
    (\cref{fig:vqaheatmaps}) and structured (\cref{fig:qaoaheatmaps}) state evolution.
    As the circuits vary in gate count, the depicted values are linearly interpolated
    to 100 steps to achieve uniform value distributions over the relative circuit
    depth accross all circuits.%
}
\label{fig:heatmaps}
\end{figure*}

General state evolution algorithms usually are quite high level from an algorithmic
standpoint. The logical structure of problem instances usually is encoded in a Hamiltonian
either driving the state evolution like in quantum annealing and its gate based counterparts
(\eg, QAOA) or serving as a cost function representation expressing the solution quality,
which then can be used to optimise free parameters of a quantum circuit. In both cases,
the problem structure is quite removed from the description of the  algorithmic dynamics.
This divide between descriptive dynamics and problem structures introduces a high
level of abstraction masking the actual dynamics.

\subsection{Problem Structure in State Evolution Ansätze}

\begin{figure}[htbp]
    \centering
    \includegraphics[width=\linewidth]{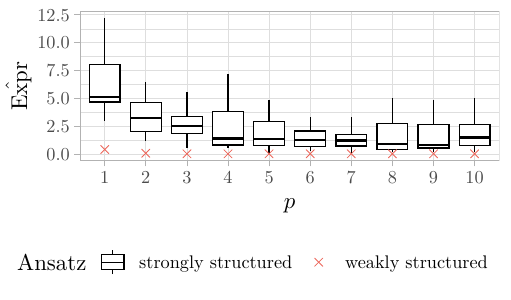}
    \caption{%
        Expressibility of a strongly structured ansatz (QAOA) and a weakly structured
        ansatz (VQE) with various layer depths $p$. In the strongly structured case,
        the ansatz is constructed with instance information. The distribution of expressibilities
        for 20 random instances per $p$ is shown as box plots. The weakly structured
        VQE ansats is instance independent and therefore is depicted as a single data
        point. The weakly structured ansatz achieves close to full expressibility
        (\ie $\expr(0)$) while the strongly structured one converges above $\expr(0)$
    }
    \label{fig:expr}
\end{figure}

When choosing an ansatz one can decide between different levels of expressiveness.
Whilst a more expressive and thus more general ansatz seems to be the obvious choice,
they usually need a high number of free parameters. Alternatively, information about
the problem can be used to construct more specific ansätze. An example for a general
ansatz would be a hardware efficient variational quantum eigensolver. In the case
of QAOA on the other hand a problem hamiltonian is driving the unitary evolution of
problem layers alternating with mixing layers. Here, the induced problem structure
restricts the state evolution of the ansatz but therefore guides the with less variational
parameters needed. 

In this work, we will use the expressibility \cite{IlrioCorrer2024,Sim2019} of an
ansatz to quantify circuit induces structure on the resulting state evolution. Following
\cite{Sim2019} we define the expressibility of a circuit by the Kullback-Leibler divergence
($\KLD$) between the fidelity distributions of pairwise sampled states from the ansatz's
parametrised output distribution and Haar randomly sampled state pairs.
\begin{definition}
    Let $C(\mathbf{\theta})$ be a parametrised ansatz, with an angle parametrisation
    $\theta \in \left[0, 2\right)^k$ for some $k \in \mathds{N}$. Given two random states $\ket|\psi>$
    and $\ket|\phi>$ the fidelity $F = \left|\braket<\psi|\phi>\right|^2$ is a random
    variable. Then, $P_C(F)$ is the probability distribution of fidelities $F(\theta,
    \varphi) = \left|\braket<\psi(\theta)|\psi(\varphi)>\right|^2$ with $\ket|\psi(\theta)>
    = C(\theta)\ket|0>$, $\ket|\psi(\varphi)> = C(\varphi)\ket|0>$ for angle parametrizations
    $\theta, \varphi$ sampled uniformly at random. Further, $P_\mathrm{Haar}(F)$ is
    the probability distribution of fidelities $F = \left|\braket<\psi|\phi>\right|^2$
    with $\ket|\psi>$ and $\ket{\phi}$ randomly sampled from the Haar distribution.
    Then, we define the expressibility of $C$ by:
$$
    \expr(C) = \KLD\ab(P_C(F) \,||\, P_\mathrm{Haar}(F))
$$
\end{definition}

If an ansatz $C$ has an expressibility $\expr(C) \approx 0$, then its output distribution
$\{C(\theta)\ketzero\}_\theta$ with regards to a random angle parametrisation $\theta$
is close to the Haar distribution. Otherwise, the structure of $C$ induces a bias
to the output state distribution forcing it to diverge from the Haar distribution.
In the picture of state evolutions we consider $\ketzero \mapsto C(\theta) \ketzero$
to be either a \emph{weakly structured state evolution} if $\expr{C} \approx 0$ or
a \emph{strongly structured state evolution} if $\expr{C} \gg 0$. 

Like \cite{Sim2019}, we will estimate $\expr(C)$ by uniformly sampling pairs of parametrisations $\theta,
\varphi \in \left[0, 2\right)^k$ and computing the corresponding fidelity $F\ab(\theta,
\varphi) = \abs{\braket<\psi(\theta)|\psi(\varphi)>}^2$, where $\ket|\psi(\cdot)>
= C(\cdot) \ketzero$. In this picture, $F$ is a random variable and we can estimate
$P_C (F)$ by the histogram of the sampled $F(\theta, \varphi)$ realisations of $F$. 
The histogram of the $F$-distribution for random Haar pairs, can be analytically computed
\cite{IlrioCorrer2024}. This gives us the estimator:
\begin{equation}
\label{eq:estExpr}
\hat{\expr}(C) = \KLD (\hat{P_C}(F) \,\|\, P_{\Haar}(F))
\end{equation}

We will now show two exemplary ansätze. One being weakly and the other strongly structured
according to our expressibility based argument.

\subsection{Weakly Structured State Evolution Ansatz}
A highly expressible (\ie weakly structured) state evolution ansatz requires a generic
circuit template leaving maximal flexibility to be adjusted later in an optimisation
step minimising a cost function which is minimal if the final state is in the target
space $\qntTrg$. The initial circuit ansatz is the same for all problems and problem
instances, otherwise problem or instance information would potentially induce circuit
structures. Concrete instance or problem specific structure only gets introduced during
the cost function optimisation process and is contained in the final optimised circuit
parametrisation $\theta \in \left[0, 2 \pi\right)$. These considerations motivate
to investigate as an exemplary ansatz a hardware efficient variational quantum eigensolver (VQE).
We chose a layered architecture where one layer exists of a stack of $R_y(\theta_i^y)$
gates applied to each qubit $i$ followed by a similar stack of $R_z(\theta_i^z)$ gates
and a ladder of cnot gates to provide entanglement. For a full circuit for the layer
structure see \cref{fig:vqeCirc}.

\Cref{fig:expr} shows the estimated expressibility $\hat{\expr}$ of the hardware efficient
VQE ansatz with different numbers of layers $p$ as described above. To compute $\hat{\expr}$
we sampled 1000 random fidelities $F(\theta, \varphi)$ and estimated $P_C$ and $P_{\Haar}$
using histograms with 100 bins of uniform bin with. One can observe that the expressiblity
of this ansatz converges close to 0, meaning that it does describe a weakly structured
state evolution.

\begin{figure}[htbp]
    \centering
    \includegraphics[width=0.7\linewidth]{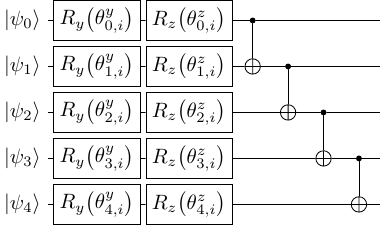}
    \caption{The $i$-th layer of the hardware efficient ansatz used for unstructured state evolution.}
    \label{fig:vqeCirc}
\end{figure}

\subsection{Strongly Structured State Evolution Ansatz}
In contrast to weakly structured state evolution techniques, in the structured case
the ansatz already gets infused with instance structures. One can show that problem
structures extrapolated from common instance structures are sufficient to successfully
approximate key parameters like the expected target space overlap of such structured
state evolution methods~\cite{Krueger2024}. This shows, that the structural infusion
significantly impacts the ansatz even before instance specific cost function optimisation
techniques are applied. As a representative for strongly structured state evolution,
we chose a standard QAOA ansatz where the driving problem Hamiltonian is the problem
Hamiltonian $H_c$ defined above. This equals the construction by Krueger and Mauerer~\cite{Krueger2024}.

For the sake of completeness we will shortly recap the construction presented in \cite{Krueger2024}.
Additionally, as it was not required in the original paper, we will also provide a
gate level decomposition of the mixer and problem unitaries. The $i$-th layer of the
described QAOA ansatz is given by 
\begin{equation}
    \label{eq:qaoaLayer}
    U_L(\beta_i, \gamma_i) = e^{-i \beta_i \mathrm{H_X}} e^{-i \gamma_i \prbHam} ,
\end{equation}
where $\mathrm{H_x} = \sum_{j = 1}^n \sigma_j^x$. Here $\sigma_j^x$ is the Pauli-$X$
operator on the $j$-th qubit. The resulting mixing layer $e^{-i \beta_i \mathrm{H_X}}$
can thus be decomposed into single qubit rotating gates
$$
e^{-i \beta_i \mathrm{H_X}} = \prod_{j = 1}^n R_x^j(\beta_i) ,
$$
where $R_x^j$ is a $x$-rotation applied to the $j$-th qubit. To decompose $e^{-i \gamma_i
\prbHam}$, we need to understand its effect. Being a projector onto the target space
$\qntTrg$, the constraint Hamiltonian is diagonal $\prbHam = \diag(c(x_0), c(x_1),
\cdots, c(x_{2^n -1}))$ with $x_i \in \bitFldn$ being the binary representation of
$i$ and $c(x_i) \in \bitFld$ the constraint function which is only 1 if $x_i$ is a
valid solution. Thus, $e^{-i \gamma_i \prbHam}$ is a phase gate which maps $\kett
\mapsto e^{-i \gamma_i} \kett$ iff $t \in T$. To decompose $e^{-i \gamma_i
\prbHam}$ we construct a subcircuit $U_C (\gamma) = \prod_{t \in T} P(t,\gamma)$ where
$P(t, \gamma) \ket|x> = e^{-i \gamma \mathbf{1}_{x = t}} \ket|x>$, with $\mathbf{1}_{x
= t} \in \bitFld$ being the indicator function which equals 1 iff $x = t$. Further,
$P(t, \gamma)$ can be implemented by applying a $X$ gate to each qubit $i \leq n - 1$ if the
$i$-th bit $t_i$ in $t$ equals 0. After that a $n-1$ multi qubit controlled phase
gate $C_{n-1} P(\gamma)$ is applied to the $n$-th qubit. Finally, the $X$ gates are
being reversed, resulting in:
\begin{equation}
    P(t, \gamma) = \prod_{i = 1}^{n - 1} X_i^{1 - t_i} C_{n - 1} P(\gamma) \prod_{i = 1}^{n - 1} X_i^{1 - t_i}
\end{equation}
Note that the $X$ gates on the $n$-th qubit can be omitted as the phase gate only
applies to the 1 state of the qubit. \Cref{fig:phasetcirc} shows the construction
of the $P(t, \gamma)$ gate. In summary, the $p$-layer QAOA ansatz used in our simulations is constructed as follows:
\begin{equation}
    U_L(\beta_p, \gamma_p) \cdots U_L(\beta_2, \gamma_2) U_L(\beta_1, \gamma_1), 
\end{equation}
$U_L$ as described by \cref{eq:qaoaLayer} with its decomposition as described above.

\begin{figure}[htbp]
    \centering
    \includegraphics{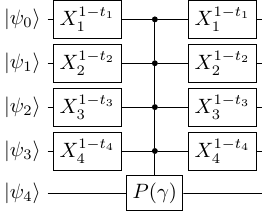}
    \caption{The decomposition of the $P(t, \gamma)$ gate.}
    \label{fig:phasetcirc}
\end{figure}

Just as in the weakly structured case we utilize the expressibility to argue that
this ansatz is strongly structured. The first argument we make is, that its expressiblity
depends on the specific problem instance. This shows, that instance information is
used to induce structure into the circuit. Further, \cref{fig:expr} shows that
even with more layers the median expressiblity converges at a higher value than the
weakly structured ansatz. As the box plots show, it also converges signifficantly
above 0 ($\hat{\expr} \gg 0$). Therfore, we argue that the ansatz represents strongly
structured state evolution ansatzes.

\subsection{Problem Description}
 
We now want to showcase our methods introduced above to reveal actual differences
in the evolution of structured and unstructured state evolution techniques. As an
exemplary problem, we chose the seminal NP complete problem of boolean satisfiability
(SAT), more precisely the problem of finding a satisfying variable assignment of a
3-CNF boolean formula $F : \bitFldn \to \bitFld$. Let's define a problem Hamiltonian
satisfying \cref{def:prbHam}. We start by defining the classical solution space $T$
where $t = t_1 t_2 \cdots t_n \in T$ iff $c\ab(t) := F\ab(t) = 1$. Then target space
shall be defined as $\qntTrg = \{\bigotimes_{i = 1}^n \ket|t_i> : t_1 t_2 \cdots t_n
\in T\}$. Note that $F$ is a 3-CNF boolean formula, therefore $F = \prod_{i = 1}^m
f_i$ with $f_i : \mathds{F}_2^3 \to \bitFld$ are disjunctions. This means, every $f_i$
has one unique unsatisfying assignment $\overline{t_i}$. Now it is easy to see that
the Hamiltonian $\prbHam := \idm - \prod_{i = 1}^m \ketbra|t_i><t_i|$ satisfies \cref{eq:prbHam}.
See Ref.~\cite{Krueger2024} for more details.

\subsection{Setup of Simulations} For each ansatz we solved 20 SAT instances with the circuits spanning $n = 7$ qubits and $p = 7$ layers. Every instance was randomly sampled with a clause to variable ratio of $\abs{C} / \abs{V} = 3$, which generates SAT instances that are constrained enough to be at the start of the easy to hard phase transition. At the same time those instances are still not too hard to solve such that we can expect the state evolutions to get fairly close to the target space, assuring that we witness a state space traversal travelling a significant part of the distance necessary to successfully solve the problem. For our simulations we used the QuTiP library \cite{lambert2024qutip5quantumtoolbox}.

\subsection{Results}
Comparing the state evolution of strongly structured and weakly structured circuits
we notice that the former apart from local fluctuations generally approaches the target
space $\ab[\qntTrg]$ in a more direct path, smoothly reducing the geodesic distance.
In contrast, the weakly structured evolution seems to more erratically move through
the state space, witnessed by bigger deviations away from the target space while passing
through the circuit. The differences become apparent when comparing the top plots
of \cref{fig:vqaheatmaps,fig:qaoaheatmaps}. Although only exemplary, we also refer
to the metric plots of \cref{fig:overview}, where the aforementioned difference becomes
especially apparent. We now further analyse how both state evolutions approached the
target space on a step by step basis. For this, we calculate the delta of $s_0$ before
and after each step. Our interest is focused on non-classical computation steps. Therefore,
we filter out all steps where no magic consumption took place. To account for numerical
errors we filter for states where $\abs{\Delta \sre} > \epsilon$, with $\epsilon =
\num{1e-5}$. Additionally, we discard outliers with a $z$-score $> 3$. The $z$-score
is defined by 
\begin{equation}
\label{eq:zscore}
z(x) = \frac{\abs{x - \mu}}{\sigma} ,
\end{equation}
where $\mu$ is the distribution mean and $\sigma$ the standard deviation.
\Cref{fig:dgd} shows that the distribution of the $\Delta s_0\ab(\ab[\qntTrg])$
is symmetrically centred around zero, in the weakly structured case. For the strongly
structured evolution, on the other hand, we observe that the distribution of $\Delta
s_0$ values is skewed towards the regime below zero (\cref{tab:delta_gd_distr} shows
more detailed numbers). This indicates that speaking on a per-step basis the strongly
structured ansatz more efficiently approaches the target, exhibiting a higher magic
efficiency. For the weakly structured ansatz the majority of steps seem to move
towards or away from the target with equal probability, de-facto cancelling each other
out on the macroscopic level. That being said, negative $\Delta s_0$ outliers (see \cref{fig:dgdOutlier}) of bigger
value seem to suggest that the weakly structured approach is able to reach further in larger
individual steps, reaching the target faster if utilised efficiently. 

\begin{table}[htbp]
    \sisetup{round-mode=places,round-precision=4}
    \centering
    \begin{tabular}{lrrrrrr}
    \toprule
        Structuredness & Q1 & Q2  & Q3 & $\Delta s_0 < -\epsilon$ & $\Delta s_0 > \epsilon$\\
     \midrule
        strongly & \num{-0.04130515} & \num{0} & \num{0} & 35.2\% & 6.01\% \\
        weakly & \num{-0.0000991} & \num{0} & \num{0} & 25.6\% & 22.8\% \\
     \bottomrule
    \end{tabular}
    \sisetup{round-mode=places,round-precision=0}
    \caption{25\% (Q1), 50\% (Q2), and 75\% (Q3) quartiles of the $\Delta s_0 \ab(\ab[\qntTrg])$
    distributions for strongly structured and weakly structured state evolution, Outliers
    with $z$-score $> 3$ were discarded. The last two columns depict fractions of
    steps that decrease ($\Delta s_0 < -\epsilon$) or increasing ($\Delta s_0 > \epsilon$)
    target distance, where $\epsilon = \num{1e-5}$. On the whole, the distribution
    of $\Delta s_0$ for strongly structured state evolutions exhibits a pronounced
    skewness towards negative values, in marked contrast to the weakly structured
    scenario, where it is mostly symmetrically centered about zero.}
    \label{tab:delta_gd_distr}
\end{table}

\begin{figure}[htbp]
    \centering
    \includegraphics{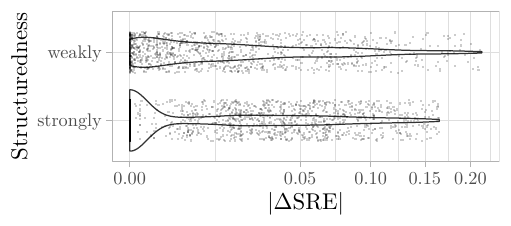}
    \caption{The distribution of magic consumption $\abs{\Delta \sre}$ of computational
    steps that had no effect on the geometric target distance $s_0 \ab(\ab[\qntTrg])$.
    Outliers with $z$-score $> 3$ were discarded. The $\abs{\Delta \sre}$ axis is square root scaled}
    \label{fig:noStepSre}
\end{figure}

Another important aspect of efficiency is the consumption of non-stabiliserness. As
already mentioned above $\sre$ is invariant under Clifford gates. In conclusion, a
change in the $\sre$ of the intermediate state being evolved indicates the use of
a non-Clifford operation. Therefore, we will use the absolute step $\sre$ difference
$\abs{\Delta \sre}$ as an indicator of non-stabiliserness consumption, which is inherently
linked to costly operations. Comparing \cref{fig:vqaheatmaps,fig:qaoaheatmaps} (bottom),
one sees that, the weakly structured ansatz has a significantly denser distribution
of magic consuming evolution steps compared to the strongly structured ansatz, whereas
the latter shows a higher peak magic consumption throughout its evolution steps. 
This begs the question whether this shows a higher magic efficiency of the structured
state evolution, especially taking the more uniform target appraoch of them compared
to weakly structured evolutions. To explore this question we will compare each geometric
step $\Delta s_0\ab(\ab[\mathcal{T}])$ with its corresponding magic consumption $\abs{\Delta
\sre}$. \Cref{fig:dgd} suggests a bimodal distribution off $\Delta s_0 \ab(\ab[\mathcal{T}])$
values differentiating between two cases: actual geometric steps ($\Delta s_0 > \epsilon$)
and computational steps where the state did neither decrease nor increase its geometric
distance to the target $\qntTrg$ ($\abs{\Delta s_0} < \epsilon$), both up to a precision
of $\epsilon$. During the strongly structured state evolution only $23.5\%$ of the
computational steps had no effect on the geometric distance to $\qntTrg$ whereas during
the weakly structured case $42.2\%$ of the steps had no effect on $s_0$. Looking further
into it, we see that the distance invariant computational steps on average consumed
more magic in the weakly structured state evolution compared to the strongly structured
one (see \cref{fig:noStepSre}). For the second mode of computational steps with $\abs{\Delta
s_0 \ab(\ab[\qntTrg])} > \epsilon$, there is a clear positive correlation between
step-wise geodesic distance reductions to the target space and magic consumption for
the strongly structured state evolution. In contrast to that observation, there is
no such correlation for the weakly structured case, as shown in \cref{fig:dsreVSdgd}.
This further substantiates the hypothesis that the strongly structured ansatz utilises
non-stabiliser resources more efficiently.

\begin{figure}[htbp]
    \centering
    \includegraphics{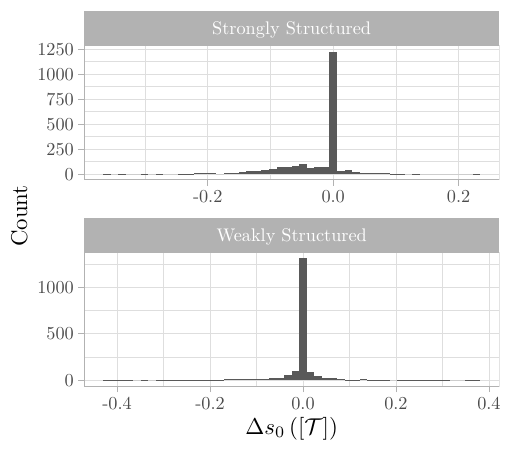}
    \caption{Distribution of increments and decrements in distance to the target space
    ($x \sim \pm \Delta s_0 \ab(\ab[\qntTrg])$) of computational steps with magic a
    consumption $\abs{\Delta \sre} > \num{1e-5}$ and $\Delta s_0$ outliers (z-score
    $>$ 3) removed. \emph{Top:} Heavy skew toward decrements is observed for the structured
    ansatz, and most distance changes are negative. \emph{Bottom:} The distribution
    of distance changes is approximately centered around zero for the unstructured
    ansatz, discounting a small number of outliers on the
    negative side.}
    \label{fig:dgd}
\end{figure}

\begin{figure}[htbp]
    \centering
    \includegraphics{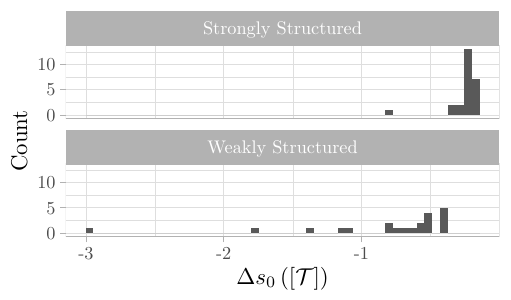}
    \caption{Distribution of negative $\Delta s_0 \ab(\ab[\qntTrg]) < -\epsilon$ outliers ($z$-score $> 3$) discarded in \cref{fig:dgd} and \cref{tab:delta_gd_distr}.}
    \label{fig:dgdOutlier}
\end{figure}

\begin{figure*}[htbp]
    \centering
    \subfloat[A close up of the filtered data points.]{\includegraphics{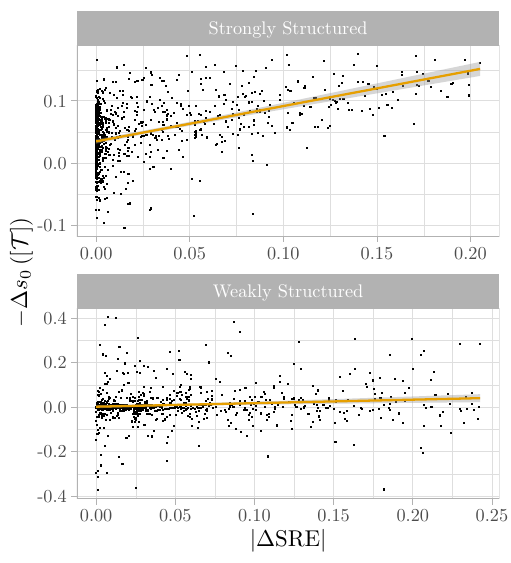}}
    \subfloat[The filtered data points in the context of all data points.]{\includegraphics{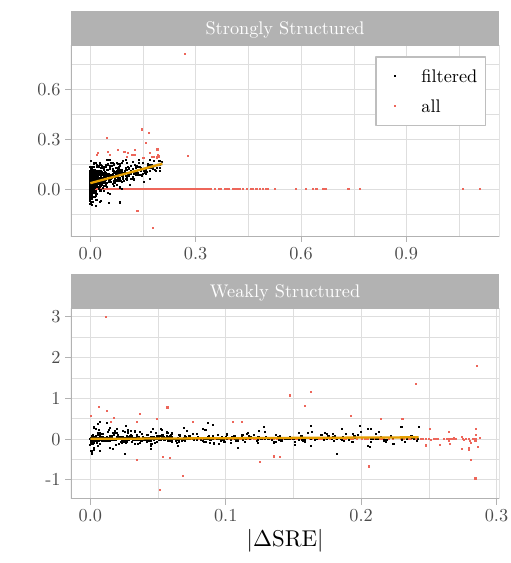}}
    \caption{%
        Correlation of magic consumption $\abs{\Delta \sre}$ and geometric target
        distance variation $\Delta s_0 \ab(\ab[\qntTrg])$ for all steps with $\Delta
        s_0 > \epsilon$: Strongly structured state evolution (top) shows clear correlation
        between magic consumption $\abs{\Delta \sre}$ and steps reducing the geodesic
        distance to the target space $- \Delta s_0 \ab(\ab[\qntTrg])$. In contrast,
        we cannot observe a similar correlation for the weakly structured ansatz (bottom).
        Outliers with $z$-score $> 3$ in $\abs{\Delta \sre}$ or $\Delta s_0 \ab(\ab[\qntTrg])$
        were discarded.} 
    \label{fig:dsreVSdgd}
\end{figure*}

\section{Discussion} 
\label{sec:discussion}
We have extended the concept of geodesic distance measures in state evolutions targeting
a specific state to evolutions targeting a more complex target space $\qntTrg$. We
showed how the geodesic distance to $\qntTrg$ can be derived from the expected value
of a Hamiltonian satisfying \cref{eq:prbHam}. This Hamiltonian based definition fits
well into widely used frameworks of Hamiltonian cost function encodings. It also further
allows for establishing empirical measurement based setups that integrating nicely
with existing toolkits of quantum computing practitioners.

We further provided a qubit order agnostic version of the geodesic framework by introducing
equivalence classes of states that are equal under permutation. Considering the quantum
Fourier transform as a use-case, we demonstrated how our approach can cut through
Clifford layers, and thus unveil previously hidden computational progress in the circuit.
We then applied the developed methods to comparatively analyse of strongly structured
and weakly structured state evolutions.

The different distributions of geodesic distance
changes in addition with a less dense distribution of magic consuming evolution steps
suggest a higher geodesic and magic efficiency for the strongly structured evolution. 

Our analysis a bimodal distribution of geodesic step sizes for weakly and strongly
structured state evolution. One mode captures all gates that have no effect on the
geometric target distance. Here magic consumption can be considered unproductive.
The weakly structured state evolution consumes more magic in these steps compared
to the strongly structured one. The second mode captures all gates that do influence
the geometric target distance, either positively or negatively. By combining resource
theoretic Stabiliser-\renyi-Entropy and geometric geodesic distance measures, we were
able to show that for these gates the strongly structured ansatz is significantly
more efficient in the consumption of non-stabiliser resources than the weakly structured
ansatz. On a methodical level, this demonstrated the potential of our combination
of methodologies. In summary, the structured evolution is less wasteful in magic consumption
with gates that don't affect the geometric target distance and additionally shows a
higher magic efficiency when actually geometrically moving the state towards the target.
Overall the structured state evolution is significantly more magic efficient. 

It is also interesting to note that the magic consumption of the structured state
evolution we investigated seems to match the observations of \cite{Capecci2025} quite
well, where magic buildup peeked roughly in the middle of the evolution before
it got reduced while closing in onto the target state. This is interesting as in \cite{Capecci2025}
the authors investigate QAOA in the context of combinatorial optimisation with low
degeneracy in the ground state (\ie target space in our terms). The authors observed
how the state has to pass a, in their terms, \emph{magic barrier} in order to reach
the non-degenerated (stabiliser) state. We argue that in a case where the initial
and target state both are low in magic, this shape is necessary to achieve quantum
speedups. This becomes clear under our perspective, where indeed not magic in itself
is the quantum resource associated with quantum computations (and thus speedups) but
rather the \emph{magic consumption} $\abs{\Delta \sre}$. Thus, the \emph{magic barrier}
observed in \cite{Capecci2025} can be seen as the accumulated magic consumption during
the state evolution as both steps increasing and decreasing magic are inherently non-classical
and can both be efficient computational steps. In fact, evolving a stabiliser initial
state to a stabiliser target state through a flat magic landscape would hint at a
primarily classic computation not touching any quantum potential. In a variational
scheme this would mean: moving the hard computational effort to the classical optimisation
loop.

In contrast to~\cite{Capecci2025} our construction captures decision problem that
quite naturally have structured solution spaces leading to entangled non-stabiliser
high magic targets spaces. Just thing of valid solutions of SAT formulas which have
inherent dependencies between singe bit flips cascading to other bits. Or in a more
pathological case, we could also force our structured ansatz to prepare the $\ket|w>$
state, which has magic of $\sre(\ket|w>) = 3 \log_2 (n) - \log_2 (7n - 6)$~\cite{Odavi2023}.
Despite this stark difference we also witnessed a \emph{magic barrier} like effect
overreaching the final magic. This raises the question if such a \emph{magic barrier}
shape is a sign of efficient quantum resource consumption in a circuit. Further, one
could analyse how the distribution between magic consumption during the state evolution
and the final state magic differs weather the quantum speedup is routed more in the
state evolution itself or the sampling from a high magic state.

Another interesting aspect probably worth a follow up investigation is the roll of
the parameter space. Intuition would suggest that the encoded problem structure of
strongly structured state evolution allows for a reduced number of well-placed parametrised
operators. In the more general weakly structured case a higher density of parameterised
gates is necessary to cover different problems and instances. On an per instance level
perhaps only an unknown subset of gates would suffice, allowing the evolution to unproductively
meander through the state space as the optimizer in a way \emph{turning the wrong
knobs}. Our framework could be utilized for instance to investigate the importance
of individual gates in weakly structured ansätze to come up with a more ansätze
for certain computational problems.

Recently, magic measures where proposed that capture non-local non-stabiliserness of
bipartite systems \cite{Korbany2025,Qian2025}. It would be interesting to
deploy our framework with non-local magic measures. However, the impact of disregarding
local magic to computational progress needs to be studied in more detail to infer
sound conclusions. In general, our framework is flexible enough to facilitate arbitrary
quantum state measures including entanglement measures.

\section{Conclusion \& Perspective}
\label{sec:perspecives}
We believe our methodology opens new means of analysing non-stabiliser effects and
the efficient utilisation of non-stabiliser resources in quantum circuits. A nuanced
understanding of such effects seems crucial for advancing the systematic development
of quantum algorithms, particularly with regard to realising quantum speed-ups in
a well-principled manner. Furthermore, we anticipate that our results will become
increasingly pertinent as the field transitions into the era of early fault-tolerant
quantum computing: In such regimes, non-stabiliser operations pose significantly greater
challenges for error correction compared to stabiliser operations. Consequently, the
use of this resource must be optimised, and we believe that our analytical framework
offers a valuable instrument in progressing towards this objective.

We showed how permutation agnostic distance measures can reveal internal non-stabiliser
effects previously hidden by a subset of Clifford operations. Our construction based
on permutation operators $\hat{\sigma}$ could be extended to accept general Clifford
operators in the sense that two states $\ket|\psi_1> \sim \ket|\psi_2>$ iff there
exists a $U \in \clfGrp$ such that $U \ket|\psi_1> = \ket|\psi_2>$. Although Clifford
circuits are classically simulable, finding minimal Clifford decompositions is not
always efficient. Thus, such an extension would need to impose some complexity theoretic
bounds on $U$ to avoid grouping states that can only be reached by overly powerful oracles.

We showed that the combination of resource theoretic and geometric tools offers a
mean to qualify resource consumptions by efficiency. We see a potential to embed quantum
resource theoretic measures like Stabiliser-\renyi-Entropies into a proper differential
geometric framework. This is a second promising avenue for improvement that would
allow us to analyse resources consumed by state evolutions following different paths
over the projective state manifold. 

\begin{appendix}
\section{Reproducibility}
A reproduction package~\cite{Mauerer2022} that contains all code required to run the
simulations in this paper is available at \cite{repPackage}\footnote{A DOI save
version will be made public with a camera ready version of the manuscript}. All
simulations are done on premise with seeded random number generators. No cloud
services were used. Thus, given the availability of the used software packages,
our simulations are fully reproducible.
\end{appendix}

\begin{acknowledgments}
This work was supported by the German Federal Ministry of Research, Technology and Space (BMFTR), funding program \emph{quantum technologies---from basic research to market}, grant number 13N17387. WM acknowledges support by the High-Tech Agenda Bavaria.
\end{acknowledgments}

\bibliography{literature}

\end{document}